\documentclass[runningheads]{llncs}

\usepackage{bbding}
\usepackage{amsfonts}
\newcommand{\E}{\mathbb{E}}
\newtheorem{mechanism}{Mechanism}
\usepackage{booktabs}
\usepackage{algpseudocode}
\usepackage{algorithm}
\usepackage{soul}

\makeatletter
\def\BState{\State\hskip-\ALG@thistlm}
\def\algbackskip{\hskip-\ALG@thistlm}

\usepackage{amsmath}

\usepackage{enumerate}

\usepackage{graphicx}
\begin{document}
\title{Cost Sharing in Two-Sided Markets}
%
%
\author{Sreenivas Gollapudi\inst{1} \and
Kostas Kollias \inst{1} \and
Ali Shameli\inst{2}}
\authorrunning{Gollapudi et al.}
%
\institute{Google Research \and
Massachusetts Institute of Technology
}
\maketitle              
\begin{abstract}
Motivated by the emergence of popular service-based two-sided markets where sellers can serve multiple buyers at the same time, we formulate and study the {\em two-sided cost sharing} problem. In two-sided cost sharing, sellers incur different costs for serving different subsets of buyers and buyers have different values for being served by different sellers. Both buyers and sellers are self-interested agents whose values and costs are private information. We study the problem from the perspective of an intermediary platform that matches buyers to sellers and assigns prices and wages in an effort to maximize gains from trade (i.e., buyer values minus seller costs) subject to budget-balance in an incentive compatible manner. In our markets of interest, agents trade the (often same) services multiple times. Moreover, the value and cost for the same service differs based on the context (e.g., location, urgency, weather conditions, etc). In this framework, we design mechanisms that are efficient, ex-ante budget-balanced, ex-ante individually rational, dominant strategy incentive compatible, and ex-ante in the core (a natural generalization of the core that we define here).

\end{abstract}
\vspace{-0.2in}
\section{Introduction}
\label{intro}

The recent emergence of sharing economy has brought renewed interest in the scientific community on studying two-sided markets where services are traded. One example of such markets are ride-sharing services like Uber and Lyft where one side of the market, i.e. drivers, provide a service to the other side of the market, namely riders. An important characteristic of such markets is the ability of a seller to offer service to {\em multiple} buyers at the same time.  For example, Uber Pool and Lyft Line typically assign a driver to multiple riders at the same time; as long as the number of riders does not exceed the capacity of the car.  This is in contrast to the one-to-one assignment that happens in other popular two-sided markets such as Amazon and Ebay. Central to the design of the above markets are the problems of price and wage computation as well as assignment of buyers to sellers.

Consider a simpler one sided case where we have multiple buyers and one service provider. In such a case, a service provider incurs a cost $c(S)$ for serving a subset $S$ of its customers. In the case of ride-sharing, $c(S)$ is the cost incurred by a cab driver to serve the riders in $S$. Each rider $i$ values the ride $v_i$ which is known only to $i$. In this case, the utility derived by the rider is $v_i - p_i$ where $p_i$ is the price charged to the rider for the ride. Depending on the pricing mechanism chosen by the ride-sharing platform, a rider might have an incentive to misreport her value to derive higher utility. The solution to this problem involves solving a {\em cost sharing} problem~\cite{DMRS18,DO17,GKL76,GKLRS07,MRS09,M99,MS01,RS09}. A cost sharing mechanism first asks each buyer to report their value for being served and then decides the assignment as well as the price each user pays on the buyer side in a way that the cost of the seller is covered by the payments of the buyers.

The reader may note that in the above {\em one-sided} setting, only the values of the users are private while the cost function $c(S)$ of the providers is known to the platform. In this study, we propose and study the {\em two-sided cost sharing} problem that generalizes the one-sided setting to the case where the costs are also private information to the sellers and the platform procures their services by offering wages. One challenge for such settings is designing a mechanism that can actually extract the true values and cost functions of buyers and sellers respectively.

In designing our mechanism, there are various objectives that we aim to achieve. A two-sided cost-sharing mechanism is efficient if it maximizes the sum of valuations of all buyers in the assignment minus the cost incurred by the sellers (which is equivalently called the {\em gains from trade}, a popular objective in the literature for designing mechanisms for two sided markets); It is dominant-strategy incentive compatible (DSIC) if for every buyer and seller, revealing their true value and cost respectively is a dominant strategy; it is weakly budget-balanced (BB) if, in the assignment, the price realized from all buyers is at least as large as the wages paid to all the sellers; it is individually rational (IR) if no agent incurs a loss participating in the mechanism; finally, a solution of the mechanism (which consists of an assignment and vectors of wages and prices) is in the {\em core} if the utilities of the agents are such that no subset of them can form a coalition and produce welfare higher than their collective utility in the proposed solution.

Two salient features of services in the sharing economy are - a) an agent participates many times in the market and b) the agent types tend to be dependent on environmental and circumstantial parameters (such as the current location, traffic volume in the surrounding area, weather conditions, urgency, etc.) and are not intrinsic to the agents. Therefore, our work focuses on designing two-sided cost sharing mechanisms that will satisfy the properties that pertain to agent utilities, namely IR and the core, in expectation. To be more precise, our mechanisms are efficient, dominant strategy IC (DSIC), ex-ante IR, ex-ante weakly BB, and ex-ante in the core. 

We note that, on top of being suitable for our applications of interest, these properties are also tight from a technical perspective: Efficiency and IC are satisfied as their strongest possible versions and weakly BB is a platform constraint that we satisfy. Strenghtening ex-ante IR is not possible even when relaxing IC to Bayesian IC (as given by the Myerson-Satterthwaite impossibility theorem \cite{MS83} for the single buyer-single seller case) or even when relaxing efficiency (gains from trade) to approximate efficiency (as shown in \cite{BD16,BM16,CGKLT17}, again for a single buyer and a single seller). Moreover, it is conjectured to be impossible even when relaxing both IC and efficiency, as supported by partial impossibility results and experimental evidence \cite{BM16}.

\vspace{-0.1 in}
\subsection{Results and Techniques}
\vspace{-0.25in}
\begin{table}
\begin{center}
\begin{tabular}{|c|c|}
\hline
\textbf{Cost function} & \textbf{Result}\\
\hline
1 submodular seller  & Optimum welfare and core\\
\hline
Capacitated NGS sellers  & Optimum welfare and core\\
\hline
General sellers with constant capacity constraints  & approx welfare and core\\
\hline
Super additive sellers  & approx welfare and core\\
\hline
\end{tabular}
\caption{Summary of our results.\label{tab}}
\end{center}
\end{table}

\vspace{-0.3in}
As we explained above, our main contributions are mechanisms that are efficient, DSIC, ex-ante weakly BB, ex-ante IR, and ex-ante in the core. In Section \ref{sec-mech} we study classes of cost functions that allow us to design an efficient mechanism, i.e. a mechanism that maximizes gains from trade. Subsequently in Section \ref{sec-approx} we study general cost functions for which we devise an approximately efficient mechanism. The cases we study, are characterized by the cost functions of the sellers. We study 4 different scenarios. First the case where we only have one single seller, second, when we have multiple sellers with negative gross substitutes cost functions, third, when we have multiple sellers with general cost functions and constant capacity constraints, and lastly, when we have multiple sellers with superadditive cost functions.

Given that our setting is multi-dimensional, it is known that the design space for truthful mechanisms is strongly restricted and the main tool in our disposal is the family of VCG-type mechanisms \cite{C71,G73,V61}. Our first technical contribution is designing a VCG-like mechanism that guarantees the DSIC property as well as the induced outcome being ex-ante in the core. To do so, we first give an algorithm that computes utilities in the core for the case with known values and costs, by means of a primal-dual pair of LPs. This result is of independent interest in itself as it generalizes a result of \cite{BHIM10} to our various models. We then show that the utilities for different realizations can be combined point-wise to yield weakly BB wages and payments that are ex-ante in the core.

A second technical contribution is the proper use of sampling to achieve our properties of interest with high probability in polynomial time in certain sub-models. We note that, in this sampling scenario, it is trickier to guarantee that the expected total utility (over sampled points) of a group of agents matches the expected welfare (over all points) that they could generate. However, with appropriate parameter selection and arguments, we show that we can approximate to arbitrary precision the utility {\em per agent}, which then yields the required properties.

Finally, for the case of multiple agents and general cost functions with constant capacity, since we can no longer achieve an assignment that maximizes gains from trade, it is more challenging to attain the DSIC property. However, we design a rounding scheme that rounds the optimal fractional solution of the LP corresponding to gains from trade, into an integral solution that achieves exactly a fixed fraction of the LP objective. This allows us to maintain the DSIC property for our mechanism. Subsequently, for the case of uncapacitated super additive cost functions, we present a transformation of the game to fit the framework of \cite{LS05} and get a truthful mechanism, even in the absence of a welfare-maximizing algorithm. The framework of \cite{LS05} is one-sided and requires utility functions that, among other properties, are monotone. Our setting is two-sided and the utility generated by each seller is not necessarily monotone in the set of buyers. However, we show how to get past these issues and design a mechanism that works for our model.

\vspace{-0.1 in}
\subsection{Related Work}

Our work is related to two areas of literature: two-sided markets and (one-sided) cost sharing. In two-sided markets, a series of papers studies two-sided auctions that approximate efficiency with respect to the sum of values of the items' holders after trade, as opposed to the gains from trade version that we study here \cite{BD16,CGKLRT17,CKLT16,DTR17,M92}. With respect to gains from trade, approximating the optimal gains from trade is impossible in many settings \cite{BD16,BM16,CGKLT17}. Nevertheless, there do exist results that impose assumptions on the distributions and/or approximate a weaker benchmark: the gains from trade achieved by the ``second-best'' mechanism of Myerson and Satterthwaite \cite{MS83}, which is known to be the one that maximizes gains from trade subject to interim IR and BIC \cite{BM16,BCWZ17,CGKLT17,M08}. All these results rely heavily on the fact that the studied settings are single-dimensional and break down even in minor departures from single-dimensionality. Approximating gains from trade even in simple multi-dimensional settings seems like a challenging problem. In a slightly different setting, other works \cite{ABGKM17,BGKM17} approximate gains from trade with non-atomic populations of agents.

On the cost sharing side of the literature, early work on cooperative games and the core can be found in \cite{B63,G59,Sc67,Sh67}. Related to our results on computing outcomes in the core for known agent types is the work by Bateni et al. \cite{BHIM10} who show how to compute a solution in the core of a game where suppliers deal with manufacturers and all information is public. With a simple transformation, one can show that the model in \cite{BHIM10} is equivalent to our model with additive cost functions. In this sense, our results on computing solutions in the core for known agent types generalize the corresponding result of \cite{BHIM10} to broader classes of cost functions. Considering mechanisms for one-sided cost sharing, most works focus on the version of the problem where there are no prior distributions on the values of the buyers.
We remark that, in the absence of prior knowledge, the two-sided setting is hopeless in trading-off efficiency and budget balance. A simple example with a buyer with a large value $v$ and a seller with $0$ cost is enough to see that. The only price and wage that make this setting efficient and truthful are a price of $0$ for the buyer and a wage of $v$ for the seller. The only exception to this rule, is the work by Fu et al. \cite{FLSS13}, who consider a Bayesian setting and show that any approximation algorithm for the underlying problem can be transformed into a mechanism for the cost sharing problem with a logarithmic loss in efficiency.

Finally, we note that our mechanisms are related to the AGV mechanism \cite{AG79} and the mechanisms in \cite{C11}, which also achieve ex-ante IR guarantees. The AGV mechanism is efficient, BB, Bayesian IC, and ex-ante IR under certain conditions (such as no costs) which are very different from our setting. The work in \cite{C11} focuses on auction settings (including a two-sided auction with a single seller and a single item) and designs mechanisms that are efficient, BB, DSIC, and ex-ante IR. Our mechanisms focus on the richer two-sided cost-sharing setting, in which we provide solutions in the core (an important consideration in our applications of interest as, otherwise, agents are incentivized to deal outside of the market), and also address computational considerations (as they run in polynomial time for various settings).

\vspace{-0.1 in}
\section{Preliminaries}\label{model}

Our market model is comprised of a set of $m\ge 1$ sellers $M$ and a set of $n\ge 1$ buyers $N$. Each buyer $i\in N$ is unit demand and has value $v_{ij}$ for being served by seller $j\in M$. Each seller $j\in M$ is endowed with a cost function $c_j(S)$ which gives the cost of the seller to serve the buyers in $S\subseteq N$. We assume that $c_j(\emptyset)=0$ for all $j\in M$. Optionally, the model can impose a capacity constraint on the sellers with each seller $j$ being able to accept $k_j$ buyers. We make the natural assumption that values are bounded and further, without loss of generality and for simplicity of exposition, that they are in $[0,1]$. Buyers and sellers interact with an intermediary platform that determines the assignment of buyers to sellers as well as prices $p_i$ for the buyers $i\in N$ and wages $w_j$ for the sellers $j\in M$. The utility of a buyer $i$ that is matched to a seller $j$ is $u_i=v_{ij}-p_i$. The utility of a seller who is assigned buyers $S$ is $u_j=w_j-c_j(S)$.

We assume the existence of discrete prior distributions over the types of each buyer and seller. The type of a seller specifies her cost function whereas the type of a buyer specifies her values. We assume agent types are drawn independently and that the prior distributions are common knowledge. Note that our discreteness assumption is a) natural, since these are distributions over possible payments which are by definition discrete, and b) without a major impact on the model since any continuous distribution can be replaced by a discrete version with an arbitrarily small approximation to the results. The solution that the platform needs to come up with is specified as an assignment of buyers to sellers, a price vector for the buyers, and a wage vector for the sellers. Throughout the paper we describe the assignment both as a collection of buyer subsets $S_1, S_2, \ldots, S_m$, with $S_j$ being the set assigned to seller $j\in M$, and as a mapping function $\sigma(\cdot)$, with $\sigma(i)$ being the seller to which buyer $i\in N$ is assigned. Generally, our goal is to maximize gains from trade, i.e., the total value of matched buyers minus the total cost of sellers. Sometimes we refer to this objective as the social welfare.

We next describe the mechanism properties that appear in our results.

\vspace{-0.15 in}
\paragraph{\textbf{Efficiency}}
A mechanism is {\em efficient} if it maximizes the gain from trade, i.e., the total value of matched buyers minus the total cost of sellers. A mechanism is {\em $\alpha$-efficient} if it achieves an $\alpha$ approximation to the optimal gains from trade.

\vspace{-0.15 in}
\paragraph{\textbf{Weak budget balance (BB)}}
A mechanism is {\em ex-ante weakly BB} if the expected sum of prices extracted from the buyers is at least equal to the expected sum of wages paid to the sellers.

\vspace{-0.15 in}
\paragraph{\textbf{Individual rationality (IR)}}
A mechanism is {\em ex-ante IR} if every agent has non-negative expected utility before all types are drawn. We say a mechanism is $\epsilon$ ex-ante IR if the agent expected utilities are at least $-\epsilon$.

\vspace{-0.15 in}
\paragraph{\textbf{Incentive compatibility}}
A solution is {\em dominant strategy incentive compatible} (DSIC) if an agent cannot improve her utility by misreporting her type, even after learning the types of other agents.

\vspace{\baselineskip}
\vspace{-0.15 in}
In addition to these standard properties, we are interested in obtaining solutions that are {\em in the core} of the cost-sharing game, something that would encourage agents to adhere to the platform's solution and stay in its market.

\vspace{-0.15 in}
\paragraph{\textbf{Cost-sharing core}}
A solution is {\em in the core} if the sum of buyer and seller utilities equals the welfare they produce and there does not exist a coalition of buyers and sellers who can generate welfare higher than the sum of their utilities. More generally, the {\em $\alpha$-core} requests that every set of agents can't produce welfare higher than $\alpha$ times their total utilities. The $\alpha$-core property guarantees that agents on platform achieve in expectation, at least $\alpha$ fraction of the maximum utility they can potentially gain in the market by forming a private coalition.

\vspace{\baselineskip}
\vspace{-0.15 in}
Before moving forward with presenting our results, we define the classes of {\em submodular} and {\em negative gross substitutes} seller cost functions which we will use as part of our results. For the latter, we begin with the standard gross substitutes definition and then present the negative gross substitutes definition and how it is placed in our framework. This class of functions is interesting to study since it represents the theoretical border of tractability for obtaining the optimal solution for this problem.

\vspace{-0.15 in}
\paragraph{\textbf{Submodular cost functions}}
A cost function $c$ is {\em submodular} if for every subsets of buyers $S,S'\in 2^N$ it is the case that $c(S)+c(S')\ge c(S\cup S')+c(S\cap S')$.


\vspace{-0.15 in}
\paragraph{\textbf{Gross substitutes functions}}
A function $f$ defined over the set of buyers $N$ satisfies the {\em gross substitutes} condition if and only if the following holds. Let $p$ be a vector of prices charged to the buyers and let $D(p)=\arg\max_{S\subseteq N}\{f(S)-\sum_{i\in S}p_i\}$ be the demand set. Then, for every price vector $p$, every $S\in D(p)$, and every $q\ge p$, there exists a set $T\subseteq N$ such that $(S\setminus A)\cup T\in D(q)$, where $A=\{i\in N:~q_i\ge p_i\}$ is the set of items for which the prices increase from $p$ to $q$.

\vspace{-0.15 in}
\paragraph{\textbf{Negative gross substitutes cost functions in two-sided cost-sharing}}
A seller cost function, which maps each subset of buyers $S\in 2^N$ to the real cost $c(S)$, satisfies the {\em negative gross substitutes} condition if and only if the following holds. Let $p$ be a vector of prices charged to the buyers and let $D(p)=\arg\max_{S\subseteq N}c(S)-\sum_{i\in S}p_i$. Then, for every price vector $p$, every $S\in D(p)$, and every $q\le p$, there exists a set $T\subseteq N$ such that $(S\setminus A)\cup T\in D(q)$, where $A=\{i\in N:~q_i\le p_i\}$.

\vspace{-0.15 in}
\paragraph{\textbf{Super additive cost functions}}
A cost function $c$ is {\em super additive} if for every pair of disjoint subsets of buyers $S,S'\in 2^N$ we have  $c(S)+c(S')\leq c(S\cup S')$.

\vspace{\baselineskip}

The class of functions satisfying the gross substitutes property contains, for example, all additive and unit-demand functions and is contained in the class of submodular functions. To abbreviate, we will say that a function is (negative) GS if it satisfies the (negative) gross substitutes condition.

In Section \ref{sec-mech}, we show how to design an efficient mechanism for special cases such as when we have one seller with a submodular cost function or the case where we have multiple sellers with negative gross substitutes cost functions. Note that, in general, this problem is very difficult, and as proved in Proposition \ref{hardness}, even for the case where the cost functions are constant over non-empty subsets, it is still NP-hard to design an efficient mechanism. That is why, in Section \ref{sec-approx}, we study approximately efficient mechanisms for more general cost functions.

\begin{proposition}\label{hardness}
	Given multiple sellers with cost functions that are constant over non-empty subsets, it is NP-hard to find an assignment that maximizes gains from trade.
\end{proposition}

We provide the proof of this Proposition in the appendix.

\vspace{-0.2 in}
\section{Efficient Mechanism}
\label{sec-mech}

In this section we describe the main mechanism for most of the settings we study. The mechanism requires access to two algorithms: a) an algorithm to compute a welfare-maximizing assignment of buyers to sellers with known values and costs and b) a deterministic algorithm to compute non-negative utilities that are in the (approximate) core of the cost-sharing game and sum up to the optimal welfare, again, with known values and costs. We discuss these algorithms further in Section \ref{sec-algs} and in the sections that correspond to the different models we study. For the polynomial time version of our mechanism, both of these algorithms must run in polynomial time. Let \textsc{welfare-alg} be the welfare maximizing algorithm and let \textsc{core-alg} be the algorithm that computes utilities in the $\alpha$-core. The exact value of $\alpha$ will depend on the exact model, i.e., on the number of sellers and the class of cost functions under consideration.

\begin{mechanism}\label{expo-time-mech}
Is defined as follows:
\begin{itemize}
\item Allocation Rule: Given the reported values $v_{ij}$ for $i\in N,j\in M$, and cost functions $c_j:2^N\to \mathbb{R}$ for $j\in M$, output the welfare-maximizing allocation computed by \textsc{welfare-alg}.
\item Pre-processing: For every realization of agent types $r$ that has some probability $q_r$, compute buyer utilities $y_i^r,i\in N,$ and seller utilities $z_j^r,j\in M$, that are non-negative, in the $\alpha$-core, and sum up to the optimal welfare using \textsc{core-alg}. Let $y_i=\sum_{r}q_r y_i^r$, be the expected utility of buyer $i$ over all realizations and let $z_j=\sum_{r}q_r z_j^r$, be the expected utility of seller $j$ over all realizations.
\item Buyer prices: The price charged to buyer $i$ is
\[ p_i=\sum_{j\in M}c_j(S_j)-\sum_{i'\in N,i'\neq i}v_{i'\sigma(i')}+\sum_{i'\in N,i'\neq i}y_{i'}+\sum_{j\in M}z_j, \]
where $S_j$ is the set of buyers assigned to seller $j$ and $\sigma(i')$ is the seller that $i'$ is assigned to.
\item Seller wages: The wage paid to seller $j$ is
\[ w_j=\sum_{i\in N}v_{i\sigma(i)}-\sum_{j'\in M,j'\neq j}c_{j'}(S_{j'})-\sum_{i\in N}y_i-\sum_{j'\in M,j'\neq j}z_{j'}, \]
where $S_j$ is the set of buyers assigned to seller $j$ and $\sigma(i)$ is the seller that $i$ is assigned to.
\end{itemize}
\end{mechanism}

\begin{theorem}\label{expo-time-thm}
Mechanism \ref{expo-time-mech} is efficient, ex-ante weakly BB, DSIC, ex-ante IR, and ex-ante in the $\alpha$-core.
\end{theorem}
\begin{proof}
We first show the mechanism is DSIC. We can then assume that the agents reveal their true types when discussing the remaining properties.  Let $\hat{v}_{ij}$ for $i\in N,j\in M$, and $\hat{c_j}(\cdot)$ for $j\in M$, be the reported values and cost functions. The utility of buyer $i$ who is matched to some seller $\sigma(i)$ will be
\[ u_i=\hat{v}_{i\sigma(i)}-p_i=\sum_{i'\in N}\hat{v}_{i'\sigma(i')}-\sum_{j\in M}\hat{c}_j(S_j)-\sum_{i'\in N, i'\neq i}y_{i'}-\sum_{j\in M}z_j. \]
The utility of seller $j$ will be
\[ u_j=w_j-c_j(S_j)=\sum_{i\in N}\hat{v}_{i\sigma(i)}-\sum_{j'\in M}\hat{c}_{j'}(S_{j'})-\sum_{i\in N}y_i-\sum_{j'\in M,j'\neq j}z_{j'}. \]
In both cases we get that the utility is the welfare shifted by a constant term that depends on the $y$ and $z$ vectors. Since the mechanism will pick the allocation that maximizes welfare, we get that it is in the best interest of the agent to reveal her true information.

Efficiency follows trivially by the fact that we output the welfare-maximizing allocation of \textsc{welfare-alg}.

To prove the mechanism is ex-ante IR we show that the expected utility of any buyer and seller is non-negative. The expected utility of a buyer $i$ is
\[ \E[u_i]=\E\left[\sum_{i'\in N}v_{i'\sigma(i')}-\sum_{j\in M}c_j(S_j)-\sum_{i'\in N,i'\neq i}y_{i'}-\sum_{j\in M}z_j\right], \]
where the $\sigma(\cdot)$ and $S_j,j=1,2,\ldots,m$, represent the optimal assignment for each realization. If we call $W$ the expected optimal welfare, we get
\[ \E[u_i]=W-\sum_{i'\in N,i'\neq i}y_{i'}-\sum_{j\in M}z_j=W-W+y_i=y_i\ge 0. \]
Similarly for a seller $j$ we get
\[ \E[u_j]=\E\left[\sum_{i\in N}v_{i\sigma(i)}-\sum_{j'\in M}c_{j'}(S_{j'})-\sum_{i\in N}y_i-\sum_{j'\in M,j'\neq j}z_{j'}\right]=z_j\ge 0. \]
This proves ex-ante IR.

Ex-ante weakly BB follows from the fact that the expected sum of utilities equals the produced welfare:
\[ \E\left[\sum_{i\in N}u_i+\sum_{j\in M}u_j\right]=\sum_{i\in N}\E[u_i]+\sum_{j\in M}\E[u_j]=\sum_{i\in N}y_i+\sum_{j\in M}z_j=W. \]
This implies the total money paid out to sellers is in expectation equal to the total money extracted from the buyers.

Switching to the ex-ante $\alpha$-core property, as we saw in the previous paragraph, our mechanism guarantees expected utility $y_i$ to each buyer $i\in N$ and expected utility $z_j$ to each seller $j\in M$. For every set $S$ that contains buyers $S^N$ and sellers $S^M$, let $W^r(S)$ be the welfare the set can produce under realization $r$. Then
\[ \alpha\left(\sum_{i\in S^N}y^r_i+\sum_{j\in S^M}z^r_j\right)\ge W^r(S), \]
by the fact that vectors $y^r,z^r$ are given by \textsc{core-alg}. Taking the weighted average of this inequality with respect to the probabilities $q_r$ shows that
\[ \alpha\left(\sum_{i\in S^N}y_i+\sum_{j\in S^M}z_j\right)\ge \E[W^r(S)], \]
which proves the ex-ante $\alpha$-core property.
\qed
\end{proof}

The second step of Mechanism 1 might not be feasible in polynomial time. We now modify our mechanism to make it run in polynomial time (assuming \textsc{welfare-alg} and \textsc{core-alg} run in polynomial time) as follows. Define \emph{Mechanism 2} to be exactly like Mechanism 1, however we replace the second step with the following:
\begin{itemize}
\item Pre-processing: Sample a set $C$ of $c=n^2(n+m)^5/\epsilon^3$ realizations of agent types, for some small parameter $\epsilon>0$. For every sample $r\in C$, compute buyer utilities $y_i^r,i\in N,$ and seller utilities $z_j^r,j\in M$, that are non-negative, in the $\alpha$-core, and sum up to the optimal welfare using \textsc{core-alg}. Let
\[ y_i=\left(\sum_{r\in C}\frac{y_i^r}{c}\right)+\frac{\epsilon}{(n+m)^2},\]
be the slightly shifted average utility of buyer $i$ over all sampled realizations and let
\[ z_j=\left(\sum_{r\in C}\frac{z_j^r}{c}\right)+\frac{\epsilon}{(n+m)^2},\]
be the slightly shifted average utility of seller $j$ over all sampled realizations.
\end{itemize}

\begin{theorem}\label{poly-time-thm}
For arbitrarily small $\epsilon>0$, Mechanism 2 runs in time polynomial in $n,m$, and $1/\epsilon$, and is efficient, DSIC, and, with probability $1-\epsilon$, ex-ante weakly BB, $\epsilon$ ex-ante IR, and ex-ante in the $\alpha(1+\delta)$-core, where $\delta=2\epsilon/W$ and $W$ is the expected optimal welfare.
\end{theorem}
\begin{proof}
It's not hard to see that the mechanism runs in polynomial time, assuming \textsc{welfare-alg} and \textsc{core-alg} run in polynomial time. Computing the allocation is a single invocation of \textsc{welfare-alg}, computing the $y_i,z_j$ variables requires a polynomial number of invocations of \textsc{core-alg} and computing the prices and wages are simple calculations.

Efficiency and the DSIC property follow exactly as in Theorem \ref{expo-time-thm}. To argue about the remaining properties, we first study the values of $y_i,i\in N$, $z_j,j\in M$, that we use in the mechanism. Let $y_i^*,i\in N$, $z_j^*,j\in M$, be the corresponding values that are produced in Mechanism \ref{expo-time-mech}, i.e., the expected values over all realizations. Since \textsc{core-alg} is deterministic, these values are fixed and well-defined. We claim that
\[ y_i\in \left[y_i^*,y_i^*+\frac{2\epsilon}{(n+m)^2}\right] \text{ with probability at least } 1-\frac{\epsilon}{n+m}.\]
Focus on a single $y_i$ and let $Y$ be a random variable equal to the average of $y_i^r$ over all samples $r\in C$. Since all values are in $[0,1]$, welfare is at most $n$ and, hence, so is every $y_i^r$. This implies the variance of any such sample is at most $n^2$, and the variance of their mean, $Y$, is at most $n^2/c$. Then, Chebyshev's inequality:
\[ \text{Pr}\left[\left|Y-y_i^*\right|\ge k\frac{n}{\sqrt{c}}\right]\le \frac{1}{k^2}, \]
for $k=\sqrt{(n+m)/\epsilon}$ and $c$ as in Mechanism 2 gives us 
$\text{Pr}\left[\left|Y-y_i^*\right|\ge \frac{\epsilon}{(n+m)^2}\right]\le \frac{\epsilon}{n+m}$.
Observe that our $y_i$ is simply $Y$ shifted by $\epsilon/(n+m)^2$, which proves our original claim. The calculation for each $z_j$ is identical. Now, we may observe that the probability that all $y_i$ and $z_j$ satisfy the condition in our claim is at least
$\left(1-\frac{\epsilon}{n+m}\right)^{n+m}\ge 1-\epsilon$.

At this point, given the fact that with high probability every $y_i$ and $z_j$ are at least (and very close to) the corresponding $y_i^*$ and $z_j^*$, we may argue about weakly BB, IR, and the core. We first observe that, with $W$ the expected optimal welfare, we get
\begin{align*}
\E\left[\sum_{i\in N}p_i-\sum_{j\in M}w_j\right]&=\E\left[(n+m-1)\left(\sum_{j\in M}c_j(S_j)-\sum_{i\in N}v_{i\sigma(i)}+\sum_{i\in N}y_i+\sum_{j\in M}z_j\right)\right]\\
&\ge (n-m+1)\left(\sum_{j\in M}\E[c_j(S_j)]-\sum_{i\in N}\E[v_{i\sigma(i)}]+\sum_{i\in N}y_i^*+\sum_{j\in M}z_j^*\right)\\
&=(n-m+1)(-W+W)=0,
\end{align*}
where the $\sigma(\cdot)$ and $S_j,j=1,2,\ldots,m$, represent the optimal assignment for each realization. This proves the mechanism is WBB.

We now focus on the expected utility of buyer $i$. We get
\begin{align*}
\E[u_i]&=\E\left[\sum_{j\in M}\sum_{i'\in S_j}v_{i'j}-\sum_{j\in M}c_j(S_j)-\sum_{i'\in N,i'\neq i}y_{i'}-\sum_{j\in M}z_j\right]\\
&=W-\sum_{i'\in N,i'\neq i}y_{i'}-\sum_{j\in M}z_j\\
&\ge W-\sum_{i'\in N,i'\neq i}\left(y_{i'}^*+\frac{2\epsilon}{(n+m)^2}\right)-\sum_{j\in M}\left(z_j^*+\frac{2\epsilon}{(n+m)^2}\right)\\
&\ge W-W+y_i^*-\frac{2\epsilon}{n+m} \ge y_i^*-\frac{2\epsilon}{n+m}\ge -\epsilon.
\end{align*}
The calculation for each seller $j$ is similar. This proves the mechanism is $\epsilon$ ex-ante IR.

We finally focus on the core property. The utilities $(y^*,z^*)$ are in the $\alpha$-core, as established in Theorem \ref{expo-time-thm}. As we saw in the previous paragraph, the expected utility of a buyer $i$ is at least $y_i^*-2\epsilon/(n+m)$ and the expected utility of each seller $j$ is at least $z_j^*-2\epsilon/(n+m)$. Then, for the total utility of any subset of buyers $S^N$ and any subset of sellers $S^M$ who generate expected welfare $W\left(S^N,S^M\right)$, we get:
\begin{align*}
U\left(S^N,S^M\right)&=\sum_{i\in S^N}\E[u_i]+\sum_{j\in S^M}\E[u_j]\\
&\ge\sum_{i\in S^N}\left(y_i^*-\frac{2\epsilon}{n+m}\right)+\sum_{j\in S^M}\left(z_j^*-\frac{2\epsilon}{n+m}\right)\\
&\ge\sum_{i\in S^N}y_i^*+\sum_{j\in S^M}z_j^*-2\epsilon\ge \frac{W\left(S^N,S^M\right)}{\alpha}-2\epsilon.
\end{align*}
Rearranging we get:
\begin{align*}
W\left(S^N,S^M\right)\le \alpha\left[U\left(S^N,S^M\right)+2\epsilon\right]\le \alpha\left[U\left(S^N,S^M\right)+2\epsilon\frac{U\left(S^N,S^M\right)}{W}\right]\le \alpha(1+\delta)U\left(S^N,S^M\right),
\end{align*}
where $\delta$ is as in the theorem statement. The second inequality follows by the fact that any subset of agents are offered expected utilities at most the expected optimal welfare which, in turn, follows by weakly BB of the mechanism. This completes our proof of the core property and the theorem.
\qed
\end{proof}

\subsection{Welfare Maximization and Core Computation Algorithms}
\label{sec-algs}

In this section we further discuss \textsc{welfare-alg} and \textsc{core-alg}. We begin with \textsc{welfare-alg} and note that in Mechanism \ref{expo-time-mech}, for which no run-time guarantees are provided, \textsc{welfare-alg} is simply exhaustive search and is available in all models. For Mechanism 2, \textsc{welfare-alg} must be an algorithm that solves the optimization problem of assigning buyers to sellers in polynomial time. This can be done in several models (as we explain in upcoming sections), such as the case with negative GS cost functions and the case with a single uncapacitated seller with a submodular cost function.

Our core computation algorithm, \textsc{core-alg}, relies on the following primal-dual pair of linear programs. The \textsc{primal}:
\begin{align*}
\text{maximize~~~~} & \sum_{j\in M}\sum_{S\subseteq N, |S|\le k_j}x_{jS}\left(\sum_{i\in S}v_{ij}-c_j(S)\right) & & \\
\text{subject to~~~~} & \sum_{S\subseteq N}x_{jS}\le 1 & & \forall j\in M\\
& \sum_{j\in M}\sum_{S\ni i}x_{jS}\le 1 & & \forall i\in N\\
& x_{jS}\ge 0 & & \forall j\in M,\forall S\subseteq N
\end{align*}
and the \textsc{dual}:
\begin{align*}
\text{minimize~~~~} & \sum_{i\in N}y_i+\sum_{j\in M}z_j & & \\
\text{subject to~~~~} & \sum_{i\in S}y_i+z_j\ge\sum_{i\in S}v_{ij}-c_j(S) & & \forall j\in M,\forall S\subseteq N, |S|\le k_j\\
& y_i\ge 0 & & \forall i\in N\\
& z_j\ge 0 & & \forall j\in M
\end{align*}
Let $W^*$ be the optimal value of \textsc{primal} and let $W$ be the optimal value among integral solutions to \textsc{primal}. The utilities that \textsc{core-alg} outputs are precisely the dual variables scaled by $W/W^*$. The following theorem shows that these values are indeed in the approximate core.
\begin{theorem}
Let $(y^*,z^*)$ be the solution to \textsc{dual} and let $(y,z)=(y^*,z^*) W/W^*$, where $W^*$ is the optimal value for \textsc{primal} and $W$ the value of the integral optimal solution to \textsc{primal}. Then $(y,z)$ gives utilities $y_i$ for the buyers $i\in N$ and utilities $z_j$ for the sellers $j\in M$ that are in the $\alpha$-approximate core, with $\alpha$ the integrality gap of \textsc{primal}.
\end{theorem}
\begin{proof}
We get
\[ \sum_{i\in N}y_i+\sum_{j\in M}z_j=\sum_{i\in N}y_i^*\frac{W}{W^*}+\sum_{j\in M}z_j^*\frac{W}{W^*}=\frac{W}{W^*}\left(\sum_{i\in N}y_i^*+\sum_{j\in M}z_j^*\right)=W, \]
which shows the sum of the agents' utilities equals the welfare they produce in the optimal assignment. This shows the first property of the core.

Next, we argue that there exists no set of agents who can form a coalition and produce welfare higher than $\alpha$ times the sum of their utilities. We first observe that it suffices to focus on coalitions with a single seller and several buyers by proving that if no such coalition (with one seller) can improve their welfare, then no other coalition (with multiple sellers) can improve their welfare either. More precisely, we argue about the contrapositive: if there exists a coalition with multiple sellers that improves its members' welfare, then there also exists a coalition with a single seller that improves its members' welfare. To see this suppose there is a coalition with multiple sellers $1,2,\ldots,l$, and consider the optimal assignment of buyers in the coalition to these sellers, in which $S_j$ is the set of buyers assigned to seller $j$. Define the improvement of the coalition as
\[ \Delta=\sum_{j=1,2,\ldots,l}\left(\sum_{i\in S_j}v_{ij}-c_j(S_j)\right)-\alpha\sum_{j=1,2,\ldots,l}\left(\sum_{i\in S_j}y_i+z_j\right), \]
Also define the improvement around each seller as
\[ \Delta_j=\left(\sum_{i\in S_j}v_{ij}-c_j(S_j)\right)-\alpha\left(z_j+\sum_{i\in S_j}y_i\right). \]
Clearly $\Delta=\sum_{j=1,2,\ldots,l}\Delta_j$, which implies that if $\Delta>0$, then there also exists some $\Delta_j>0$, which in turn implies our statement that, if there exists an improving coalition with multiple sellers, then there also exists an improving coalition with one seller.

Hence, at this point, it suffices to prove that there does not exist a coalition with one seller that can produce welfare higher than $\alpha$ times the sum of the agents' utilities in our solution. Let this coalition be seller $j$ and buyers $S$. We get:
\[ \alpha\left(\sum_{i\in S}y_i+z_j\right)=\alpha\frac{W^*}{W}\left(\sum_{i\in S}y_i^*+z_j^*\right)\ge\alpha\frac{W^*}{W}\left(\sum_{i\in S}v_{ij}-c_j(S)\right)\ge \sum_{i\in S}v_{ij}-c_j(S), \]
where the first inequality follows from the first dual constraint and the second inequality follows by the fact that $W/W^*$ is at most the integrality gap $\alpha$. This proves the second property of the $\alpha$-core.
\qed
\end{proof}

With respect to running time considerations, we need to be able to solve \textsc{dual} in polynomial time. 

For the case where we only have one submodular seller we can show that the integrality gap of \textsc{primal} is 1 and we can solve the primal optimally in polynomial time. 

This also holds when we have multiple sellers with NGS cost functions. These functions, which are strictly more general than linear functions, represent the limit of tractability for solving \textsc{primal} optimally over the space of integral solutions. Once we leave the space of NGS functions, we have to rely on approximately efficient mechanisms to achieve our results.

This implies that for these two cases, we can utilize Mechanism 2 to achieve in polynomial time, a mechanism that is efficient, DSIC, and, with probability $1-\epsilon$, ex-ante weakly BB, $\epsilon$ ex-ante IR, and ex-ante in the $(1+\delta)$-core, where $\delta=2\epsilon/W$ and $W$ is the expected optimal welfare. Please refer to \ref{sec-one} and \ref{sec-gs} for more details.



\vspace{-0.1 in}
\section{Approximately Efficient Mechanism}
\label{sec-approx}
\vspace{-0.1 in}
In this section, we present a polynomial time mechanism that addresses the case of intractable models, such as multiple submodular sellers. Our mechanism will achieve approximate efficiency and will be in the approximate core. The mechanism requires access to an algorithm that computes a convex combination of integral solutions that is equal to $1/\gamma$ times the fractional optimal solution of the \textsc{primal} of Section \ref{sec-algs}, for some given $\gamma$. We discuss this algorithm, which we call \textsc{approx-welfare-alg}, further in Section \ref{sec-alg-2}. We now present the mechanism's specifics.
\\
\vspace{-0.1 in}
\noindent\textbf{Mechanism 3.} Specifics:
\begin{itemize}
\item Allocation Rule: Given the reported values $v_{ij}$ for $i\in N,j\in M$, and cost functions $c_j:2^N\to \mathbb{R}$ for $j\in M$, let $x^*$ be the optimal solution to the \textsc{primal} linear program in Section \ref{sec-algs}. Our allocation is the lottery $x$ that is output by \textsc{approx-welfare-alg}.
\item Pre-processing: Sample a set $C$ of $c=n^2(n+m)^5/\epsilon^3$ realizations of agent types, for some small parameter $\epsilon>0$. For every sample $r\in C$, compute buyer utilities $y_i^r,i\in N,$ and seller utilities $z_j^r,j\in M$, by solving the \textsc{dual} of Section \ref{sec-algs}. Let
$y_i=\left(\sum_{r\in C}\frac{y_i^r}{c}\right)+\frac{\epsilon}{(n+m)^2}$,
be the slightly shifted average utility of buyer $i$ over all sampled realizations and let
$z_j=\left(\sum_{r\in C}\frac{z_j^r}{c}\right)+\frac{\epsilon}{(n+m)^2}$,
be the slightly shifted average utility of seller $j$ over all sampled realizations. Also, define
\[ v_i(x^*)=\sum_{j\in M, S\ni i}x^*_{jS}v_{ij} \text{~~~~and~~~~} c_j(x^*)=\sum_{S\subseteq N}x^*_{jS}c_j(S), \]
which can be interpreted as the extracted value of buyer $i$ and the incurred cost of seller $j$ under fractional solution $x^*$.
\item Buyer prices: The price charged to buyer $i$ is
\[ p_i=\frac{1}{\gamma}\left(\sum_{j\in M}c_j(x^*)-\sum_{i'\in N,i'\neq i}v_{i}(x^*)+\sum_{i'\in N,i'\neq i}y_{i'}+\sum_{j\in M}z_j\right). \]
\item Seller wages: The wage paid to seller $j$ is
\[ w_j=\frac{1}{\gamma}\left(\sum_{i\in N}v_{i}(x^*)-\sum_{j'\in M,j'\neq j}c_{j'}(x^*)-\sum_{i\in N}y_i-\sum_{j'\in M,j'\neq j}z_{j'}\right). \]
\end{itemize}

\begin{theorem}\label{thm41}
For arbitrarily small $\epsilon>0$, Mechanism 3 runs in time polynomial in $n,m$, and $1/\epsilon$, and is $\gamma$-efficient, DSIC, and, with probability $1-\epsilon$, ex-ante weakly BB, $\epsilon$ ex-ante IR, and ex-ante in the $\gamma(1+\delta)$-core, where $\delta=2\epsilon/W$ and $W$ is the expected welfare of the mechanism.
\end{theorem}
\begin{proof}
It is clear to see that the mechanism runs in polynomial time since it makes only polynomially many invocations to algorithms with polynomial running times. The fact that the mechanism is (in expectation) $\gamma$-efficient follows directly by the fact that the output is a lottery over integral solutions that is precisely a $1/\gamma$ approximation to the fractional optimal solution of \textsc{primal}.

The (expected over lottery realizations) utility of a buyer $i\in N$ is
\[ u_i=\sum_{j\in M, S\ni i}x_{jS}v_{ij}-p_i=\frac{1}{\gamma}v_i(x)-p_i=\frac{1}{\gamma}\left(W(x)+\sum_{i'\in N,i'\neq i}y_{i'}+\sum_{j\in M}z_j\right), \]
where $W(x)$ is the objective of \textsc{primal}. This is precisely the expected social welfare plus some constant terms scaled by $1/\gamma$. Then, it is in the interest of the buyer to reveal her true $v_{ij}$ values so that the mechanism will maximize her utility. Similarly, for a seller $j$ we get
\[ u_j=w_j-\sum_{S\subseteq N}x_{jS}c_j(S)=\frac{1}{\gamma}\left(W(x)+\sum_{i\in N}y_{i}+\sum_{j'\in M,j'\neq j}z_{j'}\right), \]
which proves truth-telling is an optimal strategy for each seller as well. Then we get that the mechanism is DSIC.

Now, let $y_i^*,i\in N$, and $z_j^*,j\in M$, be the average utilities given by \textsc{dual} over all possible type realizations. Similarly to the proof of Theorem \ref{poly-time-thm}, we can get that, with probability at least $1-\epsilon$, every
\[ y_i\in\left[y_i^*,y_i^*+\frac{2\epsilon}{(n+m)^2}\right] \]
and every
\[ z_j\in\left[z_j^*,z_j^*+\frac{2\epsilon}{(n+m)^2}\right]. \]
Now let $W^*$ be the expected optimal value of \textsc{primal} and \textsc{dual}. Then, we get that with probability at least $1-\epsilon$:
\begin{align*}
\E\left[\sum_{i\in N}p_i-\sum_{j\in M}w_j\right]&=\frac{1}{\gamma}(n+m-1)\left(\E\left[\sum_{j\in M}c_j(x^*)-\sum_{i\in N}v_i(x^*)\right]+\sum_{i\in N}y_i+\sum_{j\in M}z_j\right)\\
&\ge \frac{1}{\gamma}(n+m-1)\left(\E\left[\sum_{j\in M}c_j(x^*)-\sum_{i\in N}v_i(x^*)\right]+\sum_{i\in N}y^*_i+\sum_{j\in M}z^*_j\right)\\
&=(n-m+1)(-W+W)=0,
\end{align*}
which proves weakly BB.
We now focus on the expected utility of buyer $i$. We get
\begin{align*}
\E[u_i]&=\E\left[v_i(x)-\frac{1}{\gamma}\left(\sum_{i'\in N,i'\neq i}v_{i'}(x^*)-\sum_{j\in M}c_j(x^*)-\sum_{i'\in N,i'\neq i}y_{i'}-\sum_{j\in M}z_j\right)\right]\\
&=\E\left[\frac{1}{\gamma}\left(\sum_{i\in N}v_i(x^*)-\sum_{j\in M}c_j(x^*)-\sum_{i'\in N,i\neq i}y_{i'}-\sum_{j\in M}z_j\right)\right]\\
&=\frac{1}{\gamma}\left(W^*-\sum_{i'\in N,i\neq i}y_{i'}-\sum_{j\in M}z_j\right)\\
&\ge\frac{1}{\gamma}\left[W^*-\sum_{i'\in N,i'\neq i}\left(y_{i'}^*+\frac{2\epsilon}{(n+m)^2}\right)-\sum_{j\in M}\left(z_j^*+\frac{2\epsilon}{(n+m)^2}\right)\right]\\
&=\frac{1}{\gamma}\left(W^*-W^*+y_i^*-\frac{2\epsilon}{n+m}\right)\ge-\epsilon.
\end{align*}
The calculation for each seller $j$ is similar. This proves the mechanism is $\epsilon$ ex-ante IR.

We finally focus on the core property. For given subsets of buyers and sellers $S^N$ and $S^M$, let $W\left(S^N,S^M\right)$ be the welfare they can produce. We get
\begin{align*}
U(S^N,S^M)&=\sum_{i\in S^N}\E[u_i]+\sum_{j\in S^M}\E[u_j]\\
&\ge\frac{1}{\gamma}\left(\sum_{i\in S^N}y_i^*+\sum_{j\in S^M}z_j^*-2\epsilon\right)\\
&\ge\frac{W\left(S^N,S^M\right)}{\gamma}-2\epsilon.
\end{align*}
Rearranging we get:
\begin{align*}
W\left(S^N,S^M\right)\le \gamma\left[U\left(S^N,S^M\right)+2\epsilon\right]\le \gamma\left[U\left(S^N,S^M\right)+2\epsilon\frac{U\left(S^N,S^M\right)}{W}\right]\le \gamma(1+\delta)U\left(S^N,S^M\right),
\end{align*}
where $\delta$ is as in the theorem statement. The second inequality follows by the fact that any subset of agents are offered expected utilities at most the expected optimal welfare which, in turn, follows by weakly BB property of the mechanism. This completes our proof of the core property and the theorem.
\qed
\end{proof}

\vspace{-0.2 in}
\subsection{Approximate Welfare Algorithm}
\label{sec-alg-2}
In this section, we present \textsc{approx-welfare-alg}, which outputs a convex combination of integral solutions that is precisely $x^*/\gamma$ with $x^*$ the fractional optimal solution to \textsc{primal}. We study two different settings. One is for general cost functions when we treat the capacity constraint of sellers as a constant. The other setting is when the cost functions are super additive.
\subsubsection{General Cost Functions with Constant Capacity Constraints}
In this section we study the case where sellers can have general cost functions and the capacity of all the sellers is a constant $C$. One can easily extend our result to the case where these constants are different for different sellers, however, for ease of exposition, we assume the same capacity constraint for all sellers.

Our approximate welfare algorithms consists of two parts. First we need to argue that we can find an optimal feasible solution to \textsc{primal} linear program in Section \ref{sec-algs}. And next we need to argue that we can in fact assign the buyers to sellers in a way that the assignment is a convex combination of integral solutions such that the expected assignment is exactly equal to $1/\gamma$ times the fractional optimal solution for some constant $\gamma$. This equality is important, since otherwise, the proof of DSIC property in Theorem \ref{thm41} does not go through.

To optimally solve the primal, we first explain how we can design a separation oracle for the dual. This is not too hard, since for each seller $j\in M$, given the values of $\{y\}_i, \{z\}_j$, we have to be be able to check the following constraint in polynomial time

$$\sum_{i\in S}y_i+z_j\ge\sum_{i\in S}v_{ij}-c_j(S) ~~~~~~~~~~~~ \forall S\subseteq N, |S|\le C.$$

However, since $C$ is just a constant, we can just enumerate over all such subsets and check them one by one. This allows us to solve the \textsc{primal} optimally to achieve solution $x^*$. Next we will show how to round this optimal fractional solution into an integral solution, such that for each $j\in M$ and $S\subseteq N$, the probability that seller $j$ is assigned set $S$ is exactly $x^*_{j, S}/\gamma$ for $\gamma=C+1$.

To do this, first pick an arbitrary order for all pairs $(j, S)$ where $j\in M$ and $S\subseteq N$ and call them $(j_1, S_1), (j_2, S_2), \ldots, (j_k, S_k)$ where $k$ is the total number of such pairs. Then run the following rounding algorithm.

\begin{algorithm}[H]
\caption{Rounding Algorithm}
\begin{algorithmic}[1]
\State \textbf{input}: An optimal solution $x^*$ to \textsc{primal}.
\State \textbf{output}: A random integral and feasible assignment $x$ such that $\mathbb{E}[x_{j, S}]=\frac{x^*_{j, S}}{\gamma+1}$ for all $j\in M$ and $S\subseteq N$.
\State \textbf{For} i=1\ldots k \textbf{do}:
\State \quad \textbf{If} none of the buyers in $S_i$ have been assigned to a seller, and, no set of buyers have been assigned to seller $j_i$ \textbf{then}:
\State \quad \quad Define set $R$ as follows
$$R=\{l | l<i, (S_i\cap S_l \neq \emptyset \text{ or } j_l=j_i)\}$$
\State \quad \quad Assign buyers in $S_i$ to seller $j_i$ with probability $\frac{x_{j_i, S_i}}{\gamma (1-\sum_{l\in R} x_{j_l, S_l}/\gamma)}$.
\State \textbf{End}
\end{algorithmic}
\end{algorithm}

We have the following result.

\begin{theorem}\label{thm42}
	Given an optimal solution $x^*$ to \textsc{primal}, Algorithm 1 will generate a random integral assignment $x$, such that
	$$\mathbb{E}[x_{j, S}]=\frac{1}{C+1} x^*_{j, S}~~~~~~~~ \forall j\in M, S\subseteq N.$$
\end{theorem}
\begin{proof}
	To prove this, we use induction. Note that for $i=1$, the algorithm will assign the set $S_1$ to seller $j_1$ with probability $x_{j_1, S_1}/\gamma$. Now suppose this is true for all $i<k$. We will show that the algorithm will also assign set $S_k$ to seller $j_k$ with probability $\frac{x^*_{j_k, S_k}}{\gamma}$. By chain rule, and definition of Algorithm 1, we have	
\begin{align*}
	\mathbb{P}[S_k \text{ gets assigned to } j_k]=&\mathbb{P}[\text{buyers in } S_k \text{ or the seller } j_k \text{ are not assigned before iteration } k ]\times\\
& \frac{x^*_{j_i, S_i}}{\gamma (1-\sum_{l\in R} x^*_{j_l, S_l}/\gamma)}.\end{align*}
\noindent
Now by the induction hypothesis, we have
\begin{align*}
	&\mathbb{P}[\text{buyers in } S_k \text{ or the seller } j_k \text{ are not assigned before iteration } k ]= \\&1-\mathbb{P}[\text{a buyer in } S_k \text{ or the seller } j_k \text{ is assigned before iteration } k ]=(1-\sum_{l\in R} x^*_{j_l, S_l}/\gamma).
\end{align*}
\noindent
Now by putting these two equalities together, we have:
\begin{align*}
	\mathbb{P}[S_k \text{ gets assigned to } j_k]=\frac{x^*_{j_i, S_i}}{\gamma (1-\sum_{l\in R} x^*_{j_l, S_l}/\gamma)} \times (1-\sum_{l\in R} x^*_{j_l, S_l}/\gamma)=\frac{x^*_{j_i, S_i}}{\gamma}.
\end{align*}
\noindent
There is just one additional piece that we need to prove which is that for all iterations of Algorithm 1, $\sum_{l\in R} x_{j_l, S_l}/\gamma\leq 1$. To see this, note that by the first two constraints of \textsc{primal}, we have 
\begin{align*}
& \sum_{S\subseteq N}x^*_{jS}\le 1 & & \forall j\in M,\\
& \sum_{j\in N}\sum_{S\ni i}x^*_{jS}\le 1 & & \forall i\in N.
\end{align*}
\noindent Also since each set $S$ has size at most $C$, we must have for any $j\in M$ and $S\subseteq N$
\begin{align*}
	&\sum_{j'\in M} \sum_{S'\subseteq N, S'\cap S\neq \emptyset} x^*_{j', S'}\leq C,\\
	&\sum_{S'\subseteq N} x^*_{j, S'}\leq 1.
\end{align*}
\noindent
Therefore, we must have 
$$\sum_{l\in R} x_{j_l, S_l}/\gamma\leq  \sum_{j\in M} \sum_{S\subseteq N, S\cap S_i\neq \emptyset} x^*_{j, S}/\gamma + \sum_{S\subseteq N} x^*_{j_i, S}/\gamma \leq (C+1)/\gamma \leq 1$$
and this finishes the proof.
\qed
\end{proof}

\vspace{-0.2 in}
\subsubsection{Super Additive Cost Functions}
The idea in this section is to, similar to the previous section, first solve the \textsc{primal} linear program to obtain an optimal fractional solution $x^*$. This is in general a hard problem, however, we assume in this section, that we are given access to a demand oracle for the problem. The demand oracle allows us to design a separation oracle for the \textsc{dual}, which in turn allows us to obtain $x^*$ in polynomial time. We will then use a method inspired by \cite{DNS05}, that can round this fractional solution to an integral solution with a $\sqrt{n}$ approximation ratio in polynomial time. Finally, following the following lemma due to \cite{LS05}, we show how we can design an algorithm that round the solution of $x^*$ into a random integral solution $x$, such that $\mathbb{E}[x]=x^*/\sqrt{n}$. 

\begin{lemma}[Lavi and Swamy \cite{LS05}]\label{swamy}
Let $x^*$ be the fractional optimal solution to \textsc{primal} and $\gamma$ be such that there exists a $\gamma$-approximation algorithm for the buyer to seller assignment problem and $\gamma$ also bounds the integrality gap of \textsc{primal}. Then, there exists an algorithm which we call \textsc{lottery-alg} that can be used to obtain, in polynomial time, a convex combination of integral solutions that is equal to $x^*/\gamma$, under the following conditions on the welfare generated by each seller $j$ and her matched buyers: a) it is a monotone function, b) it is $0$ for an empty set of buyers, and c) we have a polynomial time demand oracle for it.
\end{lemma}
Note that the utility function of the sellers in our setting does not satisfy the conditions presented in \ref{swamy}. Namely, the welfare of each seller $V_{j}(S)=\sum_{i\in S}v_{ij}-c_j(S)$ is not monotone. To fix this issue, we define another utility function for our sellers as follows, 
\[ \hat{V}_{j}(S)=\max_{S'\subseteq S}\sum_{i\in S'}v_{ij}-c_j(S'). \]

Now we have the following result.

\begin{theorem}\label{thm43}
For the case of super additive cost functions, using \textsc{lottery-alg} of Lemma \ref{swamy}, with $\hat{V}_{j}(S)$ as the welfare functions, we can in polynomial time, achieve a random integral assignment $x$ such that $x=x^*/\gamma$, with $\gamma=O(\sqrt{n})$.
\end{theorem}
\begin{proof}
We first need to argue that \textsc{lottery-alg} can use the $\hat{V}_j(S)$ functions, which means that each one of them: a) is monotone, b) is $0$ for the empty set, and c) has a polynomial time demand oracle. Monotonicity and $\hat{V}_j(\emptyset)=0$ follow clearly by the definition of the function. The demand oracle is simply an algorithm that maximizes $V_j(\cdot)$, which can be used to calculate $\hat{V}_j(.)$ and can also be turned directly into a separation oracle for the \textsc{dual}. If solve the problem for $\hat{V}(.)$, we can easily turn that solution into a solution for our original welfare functions. To see that this works, let's suppose, for the sake of contradiction, that some maximizer $S^*$ of $V_j(\cdot)$ is not maximizing $\hat{V}_j(\cdot)$. Then, there exists some $S'$ that maximizes $\hat{V}_j(\cdot)$ and is minimal, i.e., any subset of $S'$ does not maximize $\hat{V}_{j}(S)$. Then, by definition, we get $\hat{V}_j(S^*)=V_j(S^*)$ and $\hat{V}_j(S')=V_j(S')$. This leads to a contradiction since it must be the case that $V_j(S^*)\ge V_j(S')$ and $\hat{V}_j(S')>\hat{V}_j(S^*)$.

Next, we need to prove that the solution produced by \textsc{lottery-alg} will be $x^*/\gamma$. First note that a $\gamma=O(\sqrt{n})$ algorithm for this assignment problem is given in \cite{DNS05}. We do not provide the method used in their paper here to save space. However, the main observation is that since the cost functions are super additive, if we break a set of buyers into two disjoint subsets, and randomly assign one of them to the seller, the expected cost incurred by the seller decreases. Therefore, we can take the solution $x^*$ and break it into several integral assignments while potentially breaking some subsets into smaller sets and then we randomly use one of those integral assignments.

We now show that $x^*$ is a fractional optimal solution to \textsc{primal} even if we replace the objective function coefficients $V_j(S)$ with $\hat{V}_j(S)$ to get \textsc{primal-hat}. Suppose this was not true and let $x'$ be an optimal solution to \textsc{primal-hat} that achieves objective value $W(x')>W(x^*)$ (note $W(x^*)$ is the objective value both for \textsc{primal} and \textsc{primal-hat}). We may assume every positive $x'_{jS}$ is such that $V_j(S)=\hat{V}_j(S)$, since, if this is not true, we can easily transform $x'$ into this form (move all the value of $x'_{jS}$ to the corresponding variable for the minimal set that achieves the same $\hat{V}_j(S)$). Now the objective value $W(x')$ will be the same for both \textsc{primal} and \textsc{primal-hat}, which gives a contradiction to the optimality of $x^*$ for \textsc{primal}, since it would mean that $W(x')$, which is larger than $W(x^*)$, is the objective value for \textsc{primal} as well. This means we can use $x^*$ as the optimal fractional solution to give to \textsc{lottery-alg}, together with $\hat{V}_j(\cdot)$ functions, and compute a lottery as in Lemma \ref{swamy}.

Finally, as we argued before, every integral solution in the lottery can be converted into an integral solution for the $V_j(\cdot)$ functions with the same welfare. For each seller $j$ with assigned buyer set $S$, we can compute $S^*\in\arg\max_{S'\subseteq S}\sum_{i\in S'}v_{ij}-c_j(S')$ using our demand oracle. We then keep the players in $S^*$ assigned to the seller and remove the players in $S\setminus S^*$.
\qed
\end{proof}

\bibliographystyle{splncs04}
\bibliography{refs}

\appendix
\newpage
\section{Deferred Proofs}

\subsection{Proof of Proposition \ref{hardness}}
	We prove this by using a reduction from the set cover problem. Suppose we are given an instance of the set cover problem as follows:
	\begin{itemize}
		\item Universe $U=\{1, \ldots, n\}$
		\item Collection of subsets $\mathcal{S}=\{S_1, \ldots, S_m\}$ such that $S_1\cup S_2\cup \ldots \cup S_m=U$.
	\end{itemize} 
	We want to know what's the minimum number of subsets from $\mathcal{S}$ that we can use to cover all the elements in $U$. We reduce this problem to our two-sided market problem as follows. Suppose we have have $n$ buyers and $m$ sellers and the cost of serving any non-empty subset of buyers for any seller is exactly equal to one. We also assume for each buyer $i$, the value of being served by seller $j$ is one if and only if $i\in S_j$. In this setting, it's easy to see that the optimal value of gains from trade is exactly equal to $n-\text{OPT}$ where $\text{OPT}$ is the optimal solution to the set cover instance. This implies that we cannot design a polynomial time mechanism that is efficient for our setting even for the case when the cost functions are constant over non-empty subsets.

\subsection{Single Seller}
\label{sec-one}

We now begin discussing the results our mechanism yields in various settings. We begin with the case of one seller. In this case, there is always a subset of buyers with size at most $k_j$ that maximizes the produced welfare. This set $S$ maximizes the corresponding coefficient in the objective function (which is precisely the produced welfare) of the \textsc{primal} of Section \ref{sec-algs} and setting the corresponding $x_{jS}$ variable to $1$ is an integral optimal solution to the program. This suggests the integrality gap $\alpha$ is equal to $1$ in the statement of Theorem \ref{expo-time-thm}.

\begin{corollary}
For the case of a single seller, Mechanism \ref{expo-time-mech} is efficient, ex-ante weakly BB, DSIC, ex-ante IR, and ex-ante in the core.
\end{corollary}

We now argue that Mechanism 2 can be applied to the case of a single submodular uncapacitated seller. First of all, the argument of the previous paragraph remains: there is some welfare-maximizing subset of buyers and, similarly, the integrality gap will be $1$. Now we need to show that we can run \textsc{welfare-alg} and \textsc{core-alg} in polynomial time. With regard to \textsc{welfare-alg}, our problem involves minimizing a (non-monotone) submodular function, which is known to be solvable in polynomial time. This also gives us the integral optimal of \textsc{primal} we need in \textsc{core-alg}. All that remains is to show we can solve the \textsc{dual} of Section \ref{sec-algs} in polynomial time. This program has an exponential number of constraints, however, given a separation oracle, we can use the ellipsoid method to solve it in polynomial time. Checking the non-negativity constraints is trivial. We may rewrite the other constraint as
\[ c_j(S)-v_{ij}+\sum_{i\in S}y_i+z_j\ge 0. \]
With given $y_i,z_j$ variables, the left hand side is a submodular function, which as argued above can be minimized in polynomial time. If this minimum value is non-negative then all constraints of this form are satisfied. We then get the following corollary.
\begin{corollary}
For the case of a single seller and for arbitrarily small $\epsilon>0$, Mechanism 2 runs in time polynomial in $n,m$, and $1/\epsilon$, and is efficient, DSIC, and, with probability $1-\epsilon$, ex-ante weakly BB, $\epsilon$ ex-ante IR, and ex-ante in the $(1+\delta)$-core, where $\delta=2\epsilon/W$ and $W$ is the expected optimal welfare.
\end{corollary}

\subsection{Negative Gross Substitutes Cost Functions}
\label{sec-gs}

We now study the case with capacitated NGS sellers. If the functions $c_j(.)$ are NGS, then \textsc{primal} can be reduced to maximizing a function satisfying the GS property subject to allocation constraints. It is known that in this case, the optimal solution to the linear program will be integral. This implies that again $\alpha=1$. We then get the following corollary.

\begin{corollary}
For the case of capacitated NGS sellers, Mechanism \ref{expo-time-mech} is efficient, ex-ante weakly BB, DSIC, ex-ante IR, and ex-ante in the core.
\end{corollary}

The solution to \textsc{primal} also gives us our polynomial time \textsc{welfare-alg} and our integral solution to \textsc{primal} for \textsc{core-alg}. Then, for Mechanism 2, it remains to give a separation oracle for \textsc{dual}. Rewriting the first family of constraints as we did in Section \ref{sec-one}, we get an NGS function on the left hand side. As we observed, the function can be minimized in polynomial time (since it is equivalent to GS maximization). If the minimum is non-negative then all constraints of this form are satisfied. Checking the non-negativity constraints is trivial.

\begin{corollary}
For the case of capacitated NGS sellers and for arbitrarily small $\epsilon>0$, Mechanism 2 runs in time polynomial in $n,m$, and $1/\epsilon$, and is efficient, DSIC, and, with probability $1-\epsilon$, ex-ante WBB, $\epsilon$ ex-ante IR, and ex-ante in the $(1+\delta)$-core, where $\delta=2\epsilon/W$ and $W$ is the expected optimal welfare.
\end{corollary}

\end{document}